\documentclass{article}
\usepackage[a4paper, total={6in, 8in}]{geometry}
\usepackage[utf8]{inputenc}
\usepackage{amsmath}
\usepackage{mathtools}
\usepackage{subcaption}
\usepackage{graphicx}
\usepackage{amsthm,amssymb}
\usepackage{bbm}
\usepackage{mathdots}
\usepackage{amsfonts}
\usepackage{color}
\usepackage{authblk}

\theoremstyle{definition}

\newtheorem{lemma}{Lemma}[section]

\newtheorem{proposition}{Proposition}
\allowdisplaybreaks

\title{Screening of Informed and Uninformed Experts}
\author[1, 2]{Jorge Barreras\footnote{E-mail: fbarrer@sas.upenn.edu This paper is based on my dissertation submitted in fulfillment of the requirements for the degree of Master of Science in the Department of
Economics at the Universidad de los Andes.}}
\author[2,3]{Alvaro Riascos}
\affil[1]{Department of Mathematics, University of Pennsylvania}
\affil[2]{Quantil Research, Bogota, Colombia}
\affil[3]{Department of Economics, Universidad de Los Andes}
\date{November 7 2016}

\begin{document}

\maketitle

\begin{abstract}
\noindent Testing the validity of claims made by self-proclaimed experts can be impossible when testing them in isolation, even with infinite observations at the
disposal of the tester. However, in a multiple expert setting it is possible to
design a contract that only informed experts accept and uninformed experts
reject. The tester can pit competing forecasts of future events against each
other and take advantage of the uncertainty experts have about the other experts' knowledge. This contract will work even when there is only a single
data point to evaluate. 
\end{abstract}

\section{Introduction}
\noindent This paper studies the relationship between two self-proclaimed experts who deliver forecasts of future events to a principal, called Alice. Alice needs a mechanism to induce informed experts to reveal their knowledge honestly and to screen uninformed experts that would deliver useless and potentially harmful forecasts.

In recent literature, there has been an active debate around the general problem of whether a decision maker can successfully evaluate the forecasts of self-proclaimed experts and screen them from uninformed ones. The seminal paper was one by \cite{foster1998asymptotic} where they
show that a calibration test can be passed by uninformed experts. In the setting of evaluating an expert in isolation the matter is pretty much settled with the elegant impossibility result by Olszewski and Sandroni \cite{olszewski2009strategic} showing that any test that does not rely on counterfactuals and can be passed by informed experts, can also be passed by an uninformed one. \footnote{ In \cite{olszewski2009strategic} a similar proof is given to show that, in a single expert setting, there are no contracts that only informed experts would accept} Dekel and Feinberg \cite{dekel2006non} show a test based on eventual counterfactuals that only informed experts can pass.

Such results motivate evaluating multiple experts at once, since the possibility of one forecast performing better than another in some metric opens new possibilities for screening. In such setting there have been mixed results regarding special cases; Al-Najjar and Weinstein \cite{al2008comparative} show that under the assumption that an informed expert is present, a comparative test of their forecasts can pick the true expert with high probability. Feinberg and Stewart \cite{feinberg2008testing} provide a test that, when restricted to certain types of forecasts (that could possibly contain the true odds), cannot be passed by uninformed experts. On the other hand, Olszewski and Sandroni \cite{olszewski2009manipulability} extend their previous results to show that two uninformed experts can independently pass a comparative test, even with infinitely many data points at the disposal of the tester.

Another approach to the problem deals with considering the incentives of potential experts to misrepresent their beliefs or honestly reveal their knowledge. Echenique and Shmaya \cite{echenique2007you} introduce the idea that a false expert might do `no harm’ if his false information does not worsen the outcome compared to a prior belief held by the tester. This idea is explored further by Sandroni \cite{sandroni2014least} to propose contracts, based on scoring rules, that incentivize informed experts to reveal the truth and uninformed experts to `do no harm’. Important progress was also made by Babaioff, Lambert et al.\cite{babaioff2011only} by using scoring rules that screened uninformed experts under certain non-convexity assumptions.

I extend Sandroni \cite{sandroni2014least} in the following way; Alice offers a contract to a set of experts that determines money transfers based on their forecasted odds of a future state of Nature and the actual observation of such state. This contract specifies transfers according to how high each forecast scores on a \textit{Brier score} \cite{brier1950verification} compared to the rival forecasts. Such contract can be designed so that it is accepted by informed experts and gives incentives to revealing the true odds, but it is rejected by uncertainty averse uninformed experts. The same contract is used in a general setting, where there is not a perfectly informed expert but rather partially informed ones, to screen for the better informed expert.

The result in this paper can be understood as the conclusion to the debate regarding the possibility of screening informed and uninformed experts. In a more general setting, when it is possible to test multiple experts at once, there is a simple contract that only informed experts would accept. Moreover, this result holds even when evaluating on a single data point.\footnote{Results are presented for two experts, but can trivially be extended to an arbitrary number of experts.}

\section{The Model}
Let $S$ be a finite set of states. Let $\Delta(S)$ be the set of probability distributions over $S$. Two experts, referred to as expert $1$ and expert $2$, deliver probabilistic forecasts $f_{1}$ and $f_{2} \in \Delta(S)$ to a tester named Alice. 

Alice creates a contract that specifies money transfers between her and each expert to elicit information. A contract is a payoff function $C: \Delta(S) \times \Delta(S) \times S \rightarrow \mathbb{R}$ whose value depends on the announced odds and the observed state. If any expert rejects the contract his payoff
is 0.\footnote{If only one expert remains then the tester is sure he’s informed. He can receive a second contract that
incentivizes honestly revealing his knowledge, like in \cite{sandroni2014least}.} Consider the behavior of expert $1,$ who is offered a contract $C_{1}$. In the case that both experts accept their respective contracts, they deliver odds $f_{1}$ and $f_{2}$. When state $s$ is observed, expert 1 receives (or gives) payoff $C_{1}\left(f_{1}, f_{2}, s\right)$.

If informed, expert 1 maximizes his expected utility conditional on the other expert's forecast. It is said that Expert 1 accepts $a$ contract if for every $f_{1}, f_{2} \in \Delta(S)$ the contract satisfies $E^{f_{1}}\left\{C_{1}\left(f_{1}, f_{2}, \cdot\right)\right\}>0,$ where $E^{f_{1}}$ is the expected value under odds $f_{1}$. That is, when revealing the truth gives him a positive payoff regardless of the other expert's forecast. Moreover, he honestly reveals his beliefs when for all $f_{2} \in \Delta(S)$ and $f^{\prime} \neq f_{1} \in \Delta(S)$
\begin{align*}
    E^{f_{1}}\left\{C_{1}\left(f_{1}, f_{2}, \cdot\right)\right\}>E^{f_{1}}\left\{C_{1}\left(f^{\prime}, f_{2}, \cdot\right)\right\},
\end{align*}

\noindent which is a property that proper scoring rules will guarantee and it ensures that informed experts will not misrepresent their beliefs. If uninformed, he evaluates his prospects using the minmax criteria as in Gilboa and Schmeidler \cite{gilboa2004maxmin}.\footnote{The uninformed expert is extremely averse to uncertainty and does not have a prior, instead, she will only accept a contract if she gets a positive payoff in her worst case scenario.} Considering both experts may announce their odds using random generators of forecasts $\xi_1 and \xi_2 \in \Delta(\Delta(S))$ (a mixed strategy) this can be stated as saying that he only accepts contract $C_1$ when there exists a random generator of forecasts $\xi_1$ such that
\begin{align*}
    \min _{f \in \Theta_{1} \atop \xi_{2} \in \Delta \Delta(S) } \int_{\Delta(S)} \int_{\Delta(S)} E^{f} C^{1}\left(f^{\prime}, f^{*}, \cdot,\right) \mathrm{d}\xi_{2}\left(f^{*}\right) \mathrm{d}\xi_{1}\left(f^{\prime}\right)>0,
\end{align*}
where $\Theta_{1}$ is a closed subset of $\Delta(S)$ that contains the forecasts expert 1 deems plausible.

\section{Main Result}
\begin{proposition} \label{prop1}
Assume that $\Theta_{1}$ contains at least two points. There exists a contract $C_{1}$ such that expert 1, if informed, will accept and will incentivize her to reveal her knowledge and, if uninformed, will reject.
\end{proposition}

The intuition of the proof is simple. Design a contract that gives a payoff proportional to the difference of the Brier score plus a small enough $\varepsilon$. Informed experts can be assured to
get paid at least $\varepsilon$ since the Brier score is maximized with the true odds. Uninformed experts get negative payments in the worst case scenario, which is when the other expert is informed and forecasts the truth with probability 1, because they cannot produce a randomized strategy that is always close to the truth. Because the value of $\varepsilon$ depends on the set $\Theta_1$, this solution does not give a single contract that screens every possible pair of uninformed experts, but rather, for each pair of experts a contract that screens informed from uninformed experts can be designed.\footnote{To design this contract, the tester can just query the experts for two outcomes that they deem plausible.}

This contract can be used in a more general setting involving partially informed experts. An expert is partially informed if he is uninformed and his set of plausible forecasts $\Theta$ takes the form $B_{\delta}\left(f^{*}\right)=\left\{f \in \Delta(S):\left\|f-f^{*}\right\| \leq \delta\right\}$ for $f^{*} \in \Delta(S)$.\footnote{A reasonable assumption is that the true odds be in $B_{\varepsilon}\left(f^{*}\right)$ but this is not necessary for the result.} An expert is said to be better
informed than another when their sets of plausible forecasts are $B_{\varepsilon_{1}\left(f_{1}\right)},$ respectively $B_{\varepsilon_{2}}\left(f_{2}\right),$
and $\varepsilon_{1}<\varepsilon_{2}$ for some pair of forecasts $f_{1}$ and $f_{2}$.

\begin{proposition} \label{prop2}
In a setting with two partially informed experts, there exists a contract such that only the better informed expert would accept.
\end{proposition}

The assumption that there are perfectly informed experts is unrealistic. Proposition \ref{prop2} shows
that even when one expert is slightly better informed than another, there is a contract that achieves perfect screening. Proposition \ref{prop1} might be regarded as a limit case of proposition \ref{prop2}.

This contract resembles the test in \cite{al2008comparative} in that it compares forecasts against each other in a way that the forecast of the better informed expert will perform better than the other forecasts. However, unlike the test in \cite{al2008comparative}, we do not make the strong assumption that the tester knows about the presence of one informed expert. For our contract to work, It is enough to assume that the experts are uncertain about each other’s type and strategy. This assumption cannot be dispensed entirely since, under this contract, uninformed experts with identical forecasts can always secure a positive payment.

Another noteworthy difference is that the test in \cite{al2008comparative} cannot guarantee that the informed expert will be picked. In contrast, the contract in this paper simply has no incentive for uninformed experts to accept it, thus, this mechanism achieves perfect screening.\footnote{
The comparative test in \cite{al2008comparative} cannot fail an uninformed expert that, by chance, produces a forecast close to the truth.} 

\section{Conclusion}
Screening informed and uninformed experts can be difficult when evaluating a single expert, however, the presence of multiple experts brings strategical uncertainty to uninformed experts
which can be exploited to design a contract that only informed experts would accept.

\section{Appendix}
\begin{lemma}\label{lemma_brier}
The Brier Score $B: \Delta(S) \times S \rightarrow \mathbb{R},$ defined as $B(f, s)=2 f(s)-\sum_{s^{\prime} \in S}\left(f\left(s^{\prime}\right)\right)^{2}-1$ is such that
\begin{align*}
E^{f}\{B(g, \cdot)\}=\|f\|_{2}^{2}-\|f-g\|_{2}^{2}-1,    
\end{align*}
\noindent where $\|\cdot\|_{2}$ denotes the $\mathcal{L}_{2}\left(\mathbb{R}\right)$ norm.
\end{lemma}

\begin{proof}
\begin{align*}
E^{f}\{B(g, \cdot)\} &=-\sum_{s \in S} f(s)\left(1-2 g(s)+\sum_{s^{\prime} \in S}\left(g\left(s^{\prime}\right)\right)^{2}\right) \\ &=\sum_{s^{\prime} \in S}\left(g\left(s^{\prime}\right)\right)^{2}+\sum_{s \in S} 2 f(s) g(s)-1 \\ &=\|f\|_{2}^{2}-\|f-g\|_{2}^{2}-1    
\end{align*}
\end{proof}

\noindent \textbf{PROOF OF PROPOSITION \ref{prop1}}
\begin{proof}
Let $B$ be the Brier Score, as previously defined, and let
$f_{x}$ and $f_{y}$ be two different elements of $\Theta_{1}$. Define the contract $C_{1}$ for expert 1 as $C_{1}\left(f_{1}, f_{2}, s\right)= B\left(f_{1}, s\right)-B\left(f_{2}, s\right)+\varepsilon$, where $\varepsilon=\frac{\|f-g\|_{2}^{2}}{2}$.\\

If expert 1 is informed, he accepts the contract because, applying Lemma \ref{lemma_brier}
\begin{align*}
    E^{f}\left\{C_{1}\left(f, f_{2}, \cdot\right)\right\}&=\left(\|f\|_{2}^{2}-\|f-f\|_{2}^{2}-1\right)-\left(\|f\|_{2}^{2}-\left\|f-f_{2}\right\|_{2}^{2}-1\right)+\varepsilon\\
    &=\left\|f-f_{2}\right\|_{2}^{2}+\varepsilon>0    
\end{align*}

Moreover, the informed expert honestly reveals the truth since $\forall f_{1} \neq f$

\begin{align*}
    \quad E^{f}\left\{C_{1}\left(f, f_{2}, \cdot\right)\right\}=\left\|f-f_{2}\right\|_{2}^{2}+\varepsilon>\left\|f-f_{2}\right\|_{2}^{2}-\left\|f-f_{1}\right\|_{2}^{2}+\varepsilon=E^{f}\left\{C_{1}\left(f_{1}, f_{2}, \cdot\right)\right\}
\end{align*}

On the other hand, if expert 1 is uninformed and forecasts using a random generator of forecasts $\xi_{1}$, we begin
by noting that his maxmin payoff is bounded above by the one obtained if the other expert forecasts the truth (which he would if he is informed).\footnote{Here we are using the assumption that an uninformed expert considers the possibility that a competing
expert might forecast the true odds. For example, if he is unsure about whether the other expert is informed
or not. This is a weaker assumption than the one in \cite{al2008comparative} where they assume that the
principal knows that a true expert is present.} Formally

\begin{align}
\min _{f \in \Theta_{1} \atop \xi_{2} \in \Delta \Delta(S)} & \iint\displaylimits_{\scriptscriptstyle \Delta(S) \Delta(S)} E^{f_{1}^{1}\left(f^{\prime}, f^{*},\right) d \xi_{2}\left(f^{*}\right) d \xi_{1}\left(f^{\prime}\right)}=\min _{f \in \Theta_{1} \atop \xi_{2} \in \Delta \Delta(S)} \iint\displaylimits_{\scriptscriptstyle \Delta(S) \Delta(S)} \left(\left\|f-f_{2}\right\|_{2}^{2-}\left\|f-f_{1}\right\|_{2}^{2}+\varepsilon\right) d \xi_{2}\left(f^{*}\right) d \xi_{1}\left(f^{\prime}\right) \nonumber\\
&\leq \min _{f \in \Theta_{1}} \int\displaylimits_{\scriptscriptstyle \Delta(S)} \varepsilon-\left\|f-f_{1}\right\|_{2}^{2} d \xi_{1}\left(f^{\prime}\right)=\min _{f \in \Theta_{1}} \int\displaylimits_{\scriptscriptstyle \Delta(S)} E^{f} C^{1}\left(f^{\prime}, f, \cdot\right) d \xi_{1}\left(f^{\prime}\right) \label{negative}
\end{align}

However, the expression in \eqref{negative} is negative for every value of $\xi_1$ because
\begin{align*}
   \min _{f \in \Theta_{1}} \int\displaylimits_{\Delta(S)} \varepsilon-\left\|f-f_{1}\right\|_{2}^{2} d \xi_{1}\left(f^{\prime}\right) &=\varepsilon-\max _{f \in \Theta_{1}} \int\displaylimits_{\Delta(S)}\left\|f-f_{1}\right\|_{2}^{2} d \xi_{1}\left(f^{\prime}\right) \\ & \leq \varepsilon-\max _{f \in \Theta_{1}}\left\|\hspace{3pt}\int\displaylimits_{\Delta(S)} f-f^{\prime} d \xi_{1}\left(f^{\prime}\right)\right\|_{2}^{2} \\ &=\varepsilon-\max _{f \in \Theta_{1}}\left\|\hspace{3pt}f-\int\displaylimits_{\Delta(S)} f^{\prime} d \xi_{1}\left(f^{\prime}\right)\right\|_{2}^{2},
\end{align*}
\noindent where the inequality comes from Jensen's inequality since $\|\cdot\|_{2}^{2}$ is convex. The function $\bar{f}:=\int_{\Delta(S)} f^{\prime} d \xi_{1}\left(f^{\prime}\right),$ satisfies $\bar{f} \in \Delta(S)$ and then, using the triangular inequality
\begin{align*}
    \varepsilon-\max _{f \in \Theta_{1}}\left\|f-\int_{\Delta(S)} f^{\prime} d \xi_{1}\left(f^{\prime}\right)\right\|_{2}^{2} &=\varepsilon-\max _{f \in \Theta_{1}}\|f-\bar{f}\|_{2}^{2} \\ & \leq \varepsilon-\frac{\left\|f_{x}-\bar{f}\right\|_{2}^{2}+\left\|f_{y}-\bar{f}\right\|_{2}^{2}}{2} \\ & \leq \varepsilon-\frac{\left\|f_{x}-f_{y}\right\|_{2}^{2}}{2}<0 
\end{align*}
so the uninformed expert never accepts the contract.
\end{proof}

\noindent \textbf{PROOF OF PROPOSITION \ref{prop2}}
\begin{proof}
Let $\Theta_{1}=B_{\varepsilon_{1}}\left(f_{1}\right)$ and $\Theta_{2}=B_{\varepsilon_{2}}\left(f_{2}\right)$ for a pair of
forecasts $f_{1}$ and $f_{2}$ in $\Delta(S) .$ Without loss of generality assume $\varepsilon_{2}>\varepsilon_{1}>0 .$ Let $B$ be the
Brier Score and let $\gamma$ be such that $\varepsilon_{2}^{2}>\gamma>\varepsilon_{1}^{2} .$ Define contracts $C_{1}$ and $C_{2}$ for experts 1 and 2
respectively as $C_{1}\left(f_{1}, f_{2}, s\right)=B\left(f_{1}, s\right)-B\left(f_{2}, s\right)+\gamma^{2}$ and $C_{2}\left(f_{2}, f_{1}, s\right)=B\left(f_{2}, s\right)-B\left(f_{1}, s\right)+\gamma^{2}$.

As before, provided that experts may use random generators of forecasts $\xi_{1}, \xi_{2} \in \Delta(\Delta(S)),$
expert $i$ will accept the contract if there exists $\xi_{i} \in \Delta(\Delta(S))$ such that

\begin{align*}
    \min _{f \in \Theta_{i} \atop \xi_{j} \in  \Delta \Delta(S)} \iint\displaylimits_{\scriptscriptstyle \Delta(S) \Delta(S)} \left(\left\|f^{\prime}-f\right\|+\left(\gamma-\left\|f-f^{*}\right\|_{2}^{2}\right)\right) \mathrm{d} \xi_{i}\left(f^{*}\right) \mathrm{d} \xi_{j}\left(f^{\prime}\right)>0
\end{align*}

It is simple to see that expert 1 accepts his contract, because he can get a positive payoff by making $\xi_{1}\left(\left\{f_{1}\right\}\right)=1 .$ Formally
\begin{align*}
    \min _{f \in \Theta_{i} \atop f_{2} \in \Delta(\Delta(S))}|| f^{\prime}-f||+\left(\gamma-\left\|f-f_{1}\right\|_{2}^{2}\right) d \xi_{2}\left(f^{\prime}\right) \geq \min _{f \in \Theta_{i}} \gamma-\left\|f-f_{1}\right\|_{2}^{2} \geq \gamma-\varepsilon_{1}^{2}>0,
\end{align*}
\noindent the first inequality coming from the fact the $\|\cdot\|$ is non-negative and the second from the fact that $\Theta_{1}=B_{\varepsilon_{1}}\left(f_{1}\right)$. Since there are two forecasts $f_{x}$ and $f_{y}$ in $\Theta_{2}$ such that
$\gamma<\varepsilon_{2}^{2} \leq \frac{\left\|f_{x}-f_{y}\right\|_{2}^{2}}{2}$, then expert 2 rejects the contract following the same argument as in the proof of Proposition \ref{prop1}.
\end{proof}
\bibliographystyle{plain}
\bibliography{biblio}
\end{document}